\documentclass[12pt, leqno]{amsart} 

\usepackage[mathcal]{euscript}
\usepackage{amssymb, amsfonts, amsthm, xy}
\usepackage{graphicx}
\usepackage{stackrel}

\usepackage{setspace}
\usepackage{tikz}
\usetikzlibrary{arrows,calc,decorations.markings,math,arrows.meta}
\usetikzlibrary{decorations.pathmorphing}
\usetikzlibrary{decorations.pathreplacing}
\usetikzlibrary{positioning}
\usetikzlibrary{shapes.geometric}
\usepgfmodule{oo}
\usepackage{pgflibraryarrows}
\usepackage{pgfplots}
\pgfplotsset{compat=1.10}
\usepgfplotslibrary{fillbetween}
\usetikzlibrary{patterns}
\usepackage{caption}
\usepackage{subcaption}
\usepackage{float}

\usepackage{pdfsync}

\xyoption{all} \CompileMatrices

\begin{document}

%%%%%%%%%%%%%%%%%%%%%%%%%%%%%%%%%%%%%   my definitions %%%%%%%%%%%%%%%%%%%

% tikz uses definitions below
\newcounter{n}
\setcounter{n}{0}
\newcounter{t}
\setcounter{t}{1}
\newcounter{d}
\setcounter{d}{2}
\newcounter{dt}
\setcounter{dt}{3}
\newcounter{nn}
\setcounter{nn}{-1}

\def\nodesize{3pt}

%%% end of tikz definitions

\def\ds{\displaystyle}
\def\arr{{\longrightarrow}}
\def\colim{\mathop{\rm colim}\nolimits}
\def\perm{\mathop{\rm perm}\nolimits}
\def\uperm{\mathop{\rm u\textnormal{-}perm}\nolimits}
\def\udet{\mathop{\rm u\textnormal{-}det}\nolimits}
\def\perfmatch{\mathop{\rm PerfMatch}\nolimits}
\def\pf{\mathop{\rm pf}\nolimits}
\def\sgn{\mathop{\rm sgn}\nolimits}

\def\intertitle#1{

\medskip

{\em \noindent #1}

\smallskip
}

\newtheorem{thm}{Theorem}[section]
\newtheorem{theorem}[thm]{Theorem}
\newtheorem{lemma}[thm]{Lemma}
\newtheorem{corollary}[thm]{Corollary}
\newtheorem{proposition}[thm]{Proposition}
\newtheorem{example}[thm]{Example}

\theoremstyle{definition}
\newtheorem{definition}[thm]{Definition}
\newtheorem{point}[thm]{}
\newtheorem{exercise}[thm]{Exercise}
\newtheorem{remark}[thm]{Remark}
\newtheorem{observation}[thm]{Observation}

\makeatletter
\let\c@equation\c@thm
\makeatother
\numberwithin{equation}{section}

%%%%%%%%%%%%%%%%%%%%%%%%%%%%%%%%%%%%%   temporary commands   %%%%%%%%%%%%%

%\doublespacing

\newcommand{\comment}[1]{}

%%%%%%%%%%%%%%%%%%%%%%%%%%%%%%%%%%%%%%%%%%%%%%%%%%%%  end of my stuff  %%%

\title{A simple division-free algorithm for computing Pfaffians}
\author{Adam J. Prze\'zdziecki}

%\address{Diana Dziewa-Dawidczyk, Warsaw University of Life Sciences---SGGW, Warsaw, Poland}
%\email{diana\_dziewa\_dawidczyk@sggw.edu.pl}
\address{Adam Prze\'zdziecki, Warsaw University of Life Sciences---SGGW, Warsaw, Poland}
\email{adam\_przezdziecki@sggw.edu.pl}

\maketitle
\begin{center}
\today
\end{center}

{\bf Abstract.}
We present a very simple algorithm for computing Pfaffians which uses no
division operations. Essentially, it amounts to iterating matrix multiplication and truncation. Its complexity, for a $2n\times 2n$ matrix, is $O(nM(n))$, where $M(n)$ is the cost of matrix multiplication. In case of a sparse matrix, $M(n)$ is the cost of the dense-sparse matrix multiplication.

The algorithm is an adaptation of the Bird algorithm for determinants. We show how to extract,  with practically no additional work, the characteristic polynomial and the Pfaffian characteristic polynomial from these algorithms.

\vspace{7pt}
{\bf Mathematics Subject Classification.} 15A15.

\vspace{7pt}
{\bf Keywords.}
Combinatorial problems,
Pfaffian,
Division-free algorithms,
Characteristic polynomial.

\vspace{15pt}

\section{Introduction}
\label{section-introduction}

Let $X=[x_{ij}]$ be an $n\times n$ matrix. Bird \cite{bird} defined
$$
  \mu(X)=\left[\begin{array}{cccc}
    -\sum_{i=2}^{n}x_{ii} & x_{12} & \ldots & x_{1n} \\
    0 & -\sum_{i=3}^{n}x_{ii} & \ldots & x_{2n} \\
    \ldots & \ldots & \ldots & \ldots \\
    0 & 0 & \ldots & -\sum_{i=n+1}^{n}x_{ii}
  \end{array}\right].
$$

Given $A=[a_{ij}]$, define $F^1_A=A$, and inductively, $F^{p+1}_A=\mu(F^p_A)A$. Compared to Bird's original notation, we increase the upper index $p$ by $1$ so that $p$ becomes equal to the degree of the entries of $F^p_A$ viewed as homogeneous polynomials in the $a_{ij}$. He proves
\begin{theorem}[Bird]\label{theorem-bird}
  The matrix $F^n_A$ is everywhere zero except for its leading entry, which equals $(-1)^{n-1}\det A$.
\end{theorem}

Let $A$, $B$ be two $2n\times 2n$ skew-symmetric matrices. Define $F^{1,0}_{AB}=A$, and inductively, $F^{p,p}_{AB}=\mu(F^{p,p-1}_{AB})B$ and $F^{p+1,p}_{AB}=\mu(F^{p,p}_{AB})A$. We prove
\begin{theorem}\label{theorem-main}
  The matrix $F^{n,n}_{AB}$ is everywhere zero except for its leading entry, which equals $-\pf A\cdot\pf B$.
\end{theorem}
Thus if we are interested in the Pfaffian of a single matrix $A$, it is enough to set $B$ to any matrix whose Pfaffian is $1$ and change the sign. In Section \ref{section-characteristic} we discuss three examples of such matrices and describe how to extract the characteristic polynomial from the Bird algorithm and the Pfaffian characteristic polynomial from our modification.

Clearly, the complexity of computing the determinant or the Pfaffian by means of Theorems $1$ or $2$ is $O(nM(n))$, where $M(n)$ is the complexity of matrix multiplication. We agree with Bird's claim that ``it is
difficult to imagine a simpler procedure for computing the determinant''. Recently, B\"ar \cite{bar} published a similarly simple algorithm for the Pfaffian and the Pfaffian characteristic polynomial, but his algorithm is not entirely division-free as it requires division by integers up to half of the dimension of the matrix. B\"ar remarks that his algorithm can be accelerated to $O(\sqrt{n}M(n))$ by using ideas from Preparata and Sarwate \cite{sqrtn}, but this is done at a significant cost to the simplicity of the algorithm.

\section{Proof of Theorem \ref{theorem-main}}
\label{section-proof}

The result can be proved by an adaptation of Rote's \cite{rote} or Bird's \cite{bird} argument. We choose the latter.
For the most part we follow Bird's notation. Given a word $\alpha$ in the alphabet $[1..2n]$ and a matrix $\ds A=[a_{ij}]_{i,j=1}^{2n}$, let $A[\alpha]$ denote the matrix $\ds[a_{\alpha_i\alpha_j}]_{i,j=1}^{k}$ where $k$ is the length of $\alpha$. Note that if $A$ is skew-symmetric then so is $A[\alpha]$ and for any $i$, $j$ we have $\pf A[ij]=a_{ij}$. An odd permutation of letters in $\alpha$ changes the sign of $\pf A[\alpha]$. The proof presented below depends only on the Laplace-type expansion for Pfaffians, (see Fulton and Pragacz \cite[Equation D.1, p. 116]{fulton}):
\begin{equation}\label{equation-pfaffian-expansion}
  \pf A[\alpha] =
    \sum\big\{(-1)^ja_{1j}\pf A[\alpha\setminus 1j]\mid j\in [2..2n]\big\}
\end{equation}
where $\alpha\setminus\beta$ is the word $\alpha$ with the symbols in $\beta$ removed.

We derive explicit formulas for the matrices $F^{p,p-1}_{AB}$ and $F^{p,p}_{AB}$ in terms of the Pfaffians of submatrices of $A$ and $B$.
Let $S_p(\alpha)$ denote the set of subsequences of $\alpha$ of length $p$. We claim that, for skew-symmetric $2n\times 2n$ matrices $A$ and $B$, the $ij$-th entry $x^{p,p-1}_{ij}$ of $F^{p,p-1}_{AB}$ and $x^{p,p}_{ij}$ of $F^{p,p}_{AB}$ are

\begin{equation}\label{equation-22}
  x^{p,p-1}_{ij}=\sum
    \big\{\pf A[i\alpha j]\cdot\pf B[\alpha]\mid
    \alpha\in S_{2p-2}(\beta_i)\big\},
\end{equation}
\begin{equation}\label{equation-23}
  x^{p,p}_{ij}=\sum
    \big\{\pf A[i\alpha]\cdot\pf B[\alpha j]\mid
    \alpha\in S_{2p-1}(\beta_i)\big\},
\end{equation}
where $\beta_i=[i+1..2n]$ and $i\alpha j$, $i\alpha$, $\alpha j$ denote the appropriate concatenations of words. Assuming (\ref{equation-23}), we see that $S_{2n-1}(\beta_i)$ is nonempty only for $i=1$, in which case it has only one element, $\beta_1$, and therefore
$$
  x^{n,n}_{11}=\pf A[1\beta_1]\cdot
  \pf B[\beta_11] =
  -\pf A[[1..2n]]\cdot
  \pf B[[1..2n]] = -\pf A\cdot\pf B.
$$
The change of sign is caused by the odd permutation of letters in $\beta_11\mapsto[1..2n]$. As in \cite{bird} we see that $x^{n,n}_{1j}=0$ for $j>1$ since $\beta_1j$ contains $j$ twice. This shows that (\ref{equation-23}) implies Theorem \ref{theorem-main}.

We prove (\ref{equation-22}) and (\ref{equation-23}) by induction. For $p=1$ we agree that the Pfaffian of a $0\times 0$ matrix is $1$ and we verify (\ref{equation-22}):
$$
  x^{1,0}_{ij}=\pf A[ij]\cdot\pf B[\varepsilon]=a_{ij}
$$
where $\varepsilon$ is the empty word. Verifying (\ref{equation-23}), we see that $S_1(\beta_i)=\{i+1,...,2n\}$ and
$$
  x^{1,1}_{ij}=\sum\big\{\pf A[ik]\cdot\pf B[kj]\mid
k\in S_1(\beta_i)\big\}=
\sum_{k=i+1}^{2n}a_{ik}b_{kj}
$$
is the $ij$-th entry of the matrix product of the above diagonal part of $A$ by $B$, as required.

We use (\ref{equation-22}) and (\ref{equation-23}) as the induction hypothesis and prove their analogues for $p$ increased by $1$.

{\bf Proof that (\ref{equation-22}) implies (\ref{equation-23})}. The diagonal of $\mu(F^{p,p-1}_{AB})$ is null since $\pf A[i\alpha j]$ in (\ref{equation-22}) is null for $i=j$.

For $i<j$ we use (\ref{equation-22}) to expand the $ij$-th entry of $F^{p,p}_{AB}=\mu(F^{p,p-1}_{AB})B$:
\begin{align}\label{equation-inductive-step-22-23-expand-x}
  x^{p,p}_{ij}
    &=\sum\big\{x^{p,p-1}_{ik}b_{kj}\mid k\in\beta_i\big\} \\
    &=\sum
    \big\{\pf A[i\alpha k]\cdot\pf B[\alpha]b_{kj}\mid
    \alpha\in S_{2p-2}(\beta_i), k\in\beta_i\big\}. \nonumber
\end{align}
In order to show that (\ref{equation-inductive-step-22-23-expand-x}) equals the right hand side of (\ref{equation-23}) we use (\ref{equation-pfaffian-expansion}) to expand the second Pfaffian in (\ref{equation-23}) along its first row/column and obtain
$$
  \pf B[\alpha j] = -\pf B[j\alpha] =
   -\sum\big\{b_{jk}\pf B[\alpha\setminus k]\cdot
   \sgn\left({\alpha\atop k(\alpha\setminus k)}\right)
   \biggm|
   k\in\alpha\big\}
$$
where $\sgn\left({\alpha\atop k(\alpha\setminus k)}\right)$ denotes the sign of the permutation which brings $k$ to the front.
Thus (\ref{equation-23}) equals
\begin{equation}\label{equation-inductive-step-22-23-expand-x-target-expanded}
  -\sum\big\{
    \pf A[ik(\alpha\setminus k)]\cdot\pf B[\alpha\setminus k]b_{jk}\mid
    \alpha\in S_{2p-1}(\beta_i), k\in\alpha\big\}.
\end{equation}

To prove that (\ref{equation-inductive-step-22-23-expand-x}) equals (\ref{equation-inductive-step-22-23-expand-x-target-expanded}) it is sufficient to show
\begin{align}\label{equation-inductive-step-22-23-compare-sides}
    \sum&
    \big\{\pf A[i\alpha k]\cdot\pf B[\alpha]b_{kj}\mid
    \alpha\in S_{2p-2}(\beta), k\in\beta\big\} \\
  = &\sum\big\{\pf A[ik(\alpha\setminus k)]\cdot
    \pf B[\alpha\setminus k]b_{kj}\mid
    \alpha\in S_{2p-1}(\beta), k\in\alpha\big\} \nonumber
\end{align}
for every word $\beta$ without repeated letters. The change of sign is due to $b_{kj}=-b_{jk}$.

For $\alpha$ of even length, we have $\pf A[i\alpha k]=\pf A[ik\alpha]$ and if $k\in\alpha$ then $\pf A[i\alpha k]=0$; therefore, (\ref{equation-inductive-step-22-23-compare-sides}) reduces to
\begin{align}\label{equation-indices}
  \{&(ik\alpha,\alpha)\mid
    \alpha\in S_{2p-2}(\beta), k\in\beta\setminus\alpha\} \\
  &= \{(ik(\alpha'\setminus k),\alpha'\setminus k)\mid
    \alpha'\in S_{2p-1}(\beta), k\in\alpha'\}. \nonumber
\end{align}

Take $\alpha'\in S_{2p-1}(\beta)$ and $k\in\alpha'$. Then $\alpha'\setminus k\in S_{2p-2}(\beta)$ and $k\in\beta\setminus\alpha$. Conversely, take $\alpha\in S_{2p-2}(\beta)$ and $k\in\beta\setminus\alpha$. Then $k$ can be inserted in a unique position in $\alpha$ to give a word $\alpha'\in S_{2p-1}(\beta)$. Furthermore, $\alpha=\alpha'\setminus k$.

{\bf Proof that (\ref{equation-23}) for $p$ implies (\ref{equation-22}) for $p+1$}. The $i$-th diagonal element of $\mu(F^{p,p}_{AB})$ is
$$
  d^{p,p}_{ii}=
  -\sum \{x^{p,p}_{kk}\mid k\in\beta_i\}.
$$
Applying (\ref{equation-23}) as induction hypothesis, we obtain
$$
  d^{p,p}_{ii}=
    -\sum
    \big\{\pf A[k\alpha]\cdot\pf B[\alpha k]\mid
    k\in\beta_i, \alpha\in S_{2p-1}(\beta_k)\big\}.
$$
Since $S_{2p}(\beta_i)=\{k\alpha\mid k\in\beta_i, \alpha\in S_{2p-1}(\beta_k)\}$ and $\pf B[\alpha k]=-\pf B[k\alpha]$ for $\alpha$ of odd length, we have
\begin{equation}\label{equation-dpp}
  d^{p,p}_{ii}=
    \sum
    \big\{\pf A[\alpha]\cdot\pf B[\alpha]\mid
    \alpha\in S_{2p}(\beta_i)\big\}.
\end{equation}
Using (\ref{equation-dpp}) we compute the $ij$-th entry of $F^{p+1,p}_{AB}=\mu(F^{p,p}_{AB})A$:
$$
  x^{p+1,p}_{ij}=
    d^{p,p}_{ii}a_{ij}+\sum\{x^{p,p}_{ik}a_{kj}\mid k\in\beta_i\}.
$$
We use (\ref{equation-dpp}) and the induction hypothesis (\ref{equation-23}) to expand the above as
\begin{align}\label{equation-inductive-step-23-22-expand-x}
  x^{p+1,p}_{ij}=
  &\sum
    \big\{\pf A[\alpha]\cdot\pf B[\alpha]a_{ij}\mid
    \alpha\in S_{2p}(\beta_i)\big\} \\
    &+
    \sum
    \big\{\pf A[i\alpha]\cdot\pf B[\alpha k]a_{kj}\mid
    \alpha\in S_{2p-1}(\beta_i), k\in\beta_i\big\}. \nonumber
\end{align}
We intend to prove that (\ref{equation-inductive-step-23-22-expand-x}) equals the right hand side of (\ref{equation-22}) for $p$ increased by $1$, which is
\begin{equation}\label{equation-inductive-step-23-22-expand-x-target}
  \sum
    \big\{\pf A[i\alpha j]\cdot\pf B[\alpha]\mid
    \alpha\in S_{2p}(\beta_i)\big\}.
\end{equation}
For $\alpha$ of even length we have $\pf A[i\alpha j]=-\pf A[ji\alpha]$ and we use (\ref{equation-pfaffian-expansion}) to expand the first Pfaffian in (\ref{equation-22}) along its first row/column to obtain
\begin{align*}
  \pf A[i\alpha j] &= -\pf A[ji\alpha] \\
  &=
  a_{ij}\pf A[\alpha]
   -\sum\big\{a_{jk}\pf A[i(\alpha\setminus k)]\cdot
   \sgn\left({i\alpha\atop ki(\alpha\setminus k)}\right)
   \biggm|
    k\in\alpha\big\}.
\end{align*}
Therefore (\ref{equation-inductive-step-23-22-expand-x-target}) is equal to
\begin{align}\label{equation-inductive-step-23-22-expand-x-target-expanded}
  \sum&\big\{a_{ij}\pf A[\alpha]\cdot\pf B[\alpha]\mid
    \alpha\in S_{2p}(\beta_i)\big\} \\
  &+\sum\big\{a_{jk}
    \pf A[i(\alpha\setminus k)]\cdot\pf B[k(\alpha\setminus k)]\mid
    \alpha\in S_{2p}(\beta_i), k\in\alpha\big\}. \nonumber
\end{align}
To prove that (\ref{equation-inductive-step-23-22-expand-x}) equals (\ref{equation-inductive-step-23-22-expand-x-target-expanded}) it is sufficient to show
\begin{align}\label{equation-inductive-step-23-22-compare-sides}
    \sum&
    \big\{\pf A[i\alpha]\cdot\pf B[k\alpha]a_{kj}\mid
    \alpha\in S_{2p-1}(\beta), k\in\beta\big\} \\
  &= \sum\big\{\pf A[i(\alpha\setminus k)]\cdot
    \pf B[k(\alpha\setminus k)]a_{kj}\mid
    \alpha\in S_{2p}(\beta), k\in\alpha\big\} \nonumber
\end{align}
for every word $\beta$ without repeated letters. Note the changes of signs $\pf B[\alpha k] = -\pf B[k\alpha]$ on the left, and $a_{kj}=-a_{jk}$ on the right.

If $k\in\alpha$ then $\pf B[k\alpha]=0$ and therefore (\ref{equation-inductive-step-23-22-compare-sides}) reduces to
\begin{align}
  \{&(i\alpha,k\alpha)\mid
    \alpha\in S_{2p-1}(\beta), k\in\beta\setminus\alpha\} \\
  &= \{(i(\alpha'\setminus k),k(\alpha'\setminus k))\mid
    \alpha'\in S_{2p}(\beta), k\in\alpha'\}. \nonumber
\end{align}

This holds by the same argument as in the proof of (\ref{equation-indices}).

\section{The characteristic and Pfaffian-characteristic polynomials}
\label{section-characteristic}

Both the Bird algorithm and our adaptation to Pfaffians allow extraction of the characteristic polynomials at negligible additional cost.
\begin{proposition}\label{proposition-det-characteristic}
The coefficients of the characteristic polynomial
  $$
  \det(\lambda I-A)=
  \lambda^n+c_1\lambda^{n-1}+c_2\lambda^{n-2}+
  \cdots+c_{n-1}\lambda+c_n
  $$
  are given by $c_p=-\mathop{\rm tr} F^p_A$.
\end{proposition}
\begin{proof}
  By \cite[Equation (1)]{bird} the $i$-th diagonal entry of $F^p_A$ is
  $$
    x_{ii}^p = -(-1)^p\sum\big\{\det A[i\alpha]
      \mid \alpha\in S_{p-1}([i+1..n])\big\};
  $$
  note that the index $p$ is increased by $1$ relative to Bird's original notation. We have
  $$
    c_p = \sum\big\{\det (-A[\alpha])
      \mid \alpha\in S_{p}([1..n])\big\}
      = -\sum\big\{x_{ii}^p\mid i\in[1..n]\big\}
      = -\mathop{\rm tr} F^p_A.
  $$
\end{proof}

In the case of the Pfaffian polynomial
$$
  \pf(\lambda B-A)
$$
the choice of the skew-symmetric matrix $B=[b_{ij}]$ which should play the role of the identity matrix is not obvious. Below we cite three examples of such matrices $B$ which are most often met in the literature or for other reasons interesting.
\begin{align*}
  &B_0: \left\{\begin{array}{ll}
         b_{2i-1,2i} = 1, & \\
         b_{2i,2i-1} = -1, & \\
         0 & \mbox{otherwise,}
       \end{array}\right. \\
  &B_1: b_{ij} = \left\{\begin{array}{ll}
         1 & \mbox{if }i<j,\\
         -1 & \mbox{if }i>j,\\
         0 & \mbox{if }i=j,
       \end{array}\right. \\
  &B_2: b_{ij} = \left\{\begin{array}{ll}
         (-1)^{j-i+1} & \mbox{if }i<j,\\
         (-1)^{i-j} & \mbox{if }i>j,\\
         0 & \mbox{if }i=j.
       \end{array}\right. \\
\end{align*}

The Pfaffian of each of these matrices is $1$. The matrix $B_0$ is sparse and it is the easiest one to guess; it is considered for example by B\"{a}r \cite{bar} and Rote \cite{rote}. From the combinatorial point of view, which is closest to the author, the matrix $B_1$ is more interesting since for every $\alpha\in S_{2p}([1..2n])$ we have $\pf B_1[\alpha]=1$. For this reason $B_1$ is considered by Rote \cite{rote} and by Iwata \cite{iwata} who also uses $B_2=-B_1^{-1}$. Knus, Merkurjev, Rost and Tignol \cite{involutions} and Krivoruchenko \cite{krivoruchenko} work with $\pf B\cdot\pf(\lambda B^{-1}-A)$ for general, invertible $B$.

We have the following lemma for the Pfaffian of the sum of $B_2$ and another skew-symmetric matrix.

\begin{lemma}[Stembridge {\cite[Lemma 4.2]{stembridge}}]
\label{lemma-sum}
  If $B_2$ is as above, then
  $$
    \pf(B_2+A) = \sum\big\{\pf A[\alpha]\mid
      \alpha\in S_{2k}([1..2n]),k\in[1..n]\big\}.
  $$

\end{lemma}

The closest analogy for Proposition \ref{proposition-det-characteristic} is
\begin{proposition} The coefficients of the Pfaffian characteristic polynomial
  $$
  \pf(\lambda B_2-A)=
  \lambda^n+c_1\lambda^{n-1}+c_2\lambda^{n-2}+
  \cdots+c_{n-1}\lambda+c_n
  $$
  are given by $c_p=(-1)^{p+1}\mathop{\rm tr} F^{p,p}_{AB_1}$.
\end{proposition}

Note that the algorithm uses $B_1$ while in the definition of the polynomial we have $B_2=-B_1^{-1}$.

\begin{proof}
Lemma \ref{lemma-sum} implies
$$
  c_p=(-1)^p\sum\big\{\pf A[\alpha]\mid
    \alpha\in S_{2p}([1..2n])\big\}.
$$

Equation (\ref{equation-23}) for the pair of matrices $A$, $B_1$ and $\pf B_1[\alpha]=1$ yields
\begin{align*}
  x^{p,p}_{ii}
    &=-\sum\big\{\pf A[i\alpha]\cdot\pf B_1[i\alpha]\mid
       \alpha\in S_{2p-1}(\beta_i)\big\}  \\
    &=-\sum\big\{\pf A[i\alpha]\mid
       \alpha\in S_{2p-1}(\beta_i)\big\},
\end{align*}
hence
$$
  c_p = (-1)^{p+1}\mathop{\rm tr} F^{p,p}_{AB_1}.
$$
\end{proof}

\end{document}